\documentclass[letterpaper, 10 pt, conference]{ieeeconf}  

\IEEEoverridecommandlockouts                              

\usepackage{amsthm}
\usepackage{graphics} 
\usepackage{epsfig} 
\usepackage{times} 
\usepackage{amsmath} 
\usepackage{amssymb}  
\usepackage[noadjust]{cite}
\usepackage{hyperref}
\usepackage{dsfont}
\usepackage{enumerate}

\newcommand{\dx}{\dot{x}}

\newcommand{\tx}{\tilde{x}}
\newcommand{\dtx}{\dot{\tilde x}}
\newcommand{\ty}{\tilde{y}}
\newcommand{\nullsp}{\mathcal{N}}
\newcommand{\range}{\mathcal{R}}
\newcommand{\dxa}{\dot{\mathbf{x}}}
\newcommand{\xa}{\mathbf{x}}
\newcommand{\ya}{\mathbf{y}}
\newcommand{\R}{\mathbb{R}}
\theoremstyle{plain}
\newtheorem{theorem}{Theorem}
\newtheorem{remark}{Remark}
\newtheorem{lemma}{Lemma}

\newtheorem{corollary}{Corollary}
\theoremstyle{definition}
\newtheorem{definition}{Definition}

\title{\LARGE \bf
Identifying Network Structure of Nonlinear Dynamical Systems: Contraction and Kuramoto Oscillators
}

\author{Jaidev Gill and Jing Shuang (Lisa) Li
\thanks{J.G. and J.S.L. are with the Department of Electrical Engineering and Computer Science, University of Michigan, Ann Arbor, MI, 48109, USA.
        {\tt\small \{jaidevg, jslisali\}@umich.edu}.    }%
}

\begin{document}

\maketitle
\thispagestyle{empty}
\pagestyle{empty}

\begin{abstract}
In this work, we study the identifiability of network structures (i.e., topologies) for networked nonlinear systems when partial measurements of the nodal dynamics are taken. 
We explore scenarios where different candidate structures can yield similar measurements, thus limiting identifiability. 
To do so, we apply the contraction theory framework to facilitate comparisons between different networks.
We show that semicontraction in the observable space is a sufficient condition for two systems to become indistinguishable from one another based on partial measurements. 
We apply this framework to study networks of Kuramoto oscillators, and discuss scenarios in which different network structures (both connected and disconnected) become indistinguishable.
\end{abstract}

\section{Introduction and Motivation}\label{sec:Intro}
Many aspects of the natural world (e.g., brains, social networks) can be described as networks of interconnected dynamical systems \cite{bullmore2009complex,MARDER_2011}. 
Thus, there is significant interest in the identification of network structures (i.e., topology) from time series measurements of nodal dynamics \cite{STEPANIANTS_2020, GREWAL_1976, CHEN_2009}. 
In the context of neuroscience, estimating the structure of neuronal circuits and brain networks from partial measurements is a longstanding problem \cite{LI_2025, MARDER_2011}. 
Moreover, given a set of measurements, different methods may identify different networks \cite{STEPANIANTS_2020, GILL_2025}.
In networked nonlinear systems, behavior such as  synchrony \cite{DORFLER_2014, JADBABAIE_2004} can cause measurements from several different systems to become identical, presenting challenges for identification \cite{CHEN_2009} (see Fig.~\ref{fig:overview}). 
Access to partial measurements of the nodal dynamics exacerbate this challenge. 

In this paper, we leverage contraction theory to compare the dynamics of networked nonlinear systems\footnote{For results on network structure identification in the linear setting, see companion paper \cite{GILL_2025}.} with different network structures.
Contraction theory provides a framework to study the dynamics of nonlinear systems with varying initial conditions, and has been used to discuss the behavior of interconnected systems and networks of oscillators \cite{WANG_2004, LOHMILLER_1998, FB-CTDS}. 
We are specifically interested in the \textit{limits} of topology identification in the nonlinear setting, that is: 
\textbf{which network structures are impossible to distinguish based on partial measurement?}
In the context of contraction, this translates to: which network structures exhibit \textit{contraction in the observable space}?
This diverges from standard notions of contraction, which assume full observation and focus on one system with varying initial conditions.
In contrast, our work focuses on contraction in the observable space, i.e., different systems exhibiting similar observed dynamics (we provide a rigorous definition in Section \ref{sec:Contraction_theory}).

\begin{figure}[t!]
      \centering
      \includegraphics[clip, trim=4.8cm 3.8cm 6.5cm 3cm,width = \linewidth]{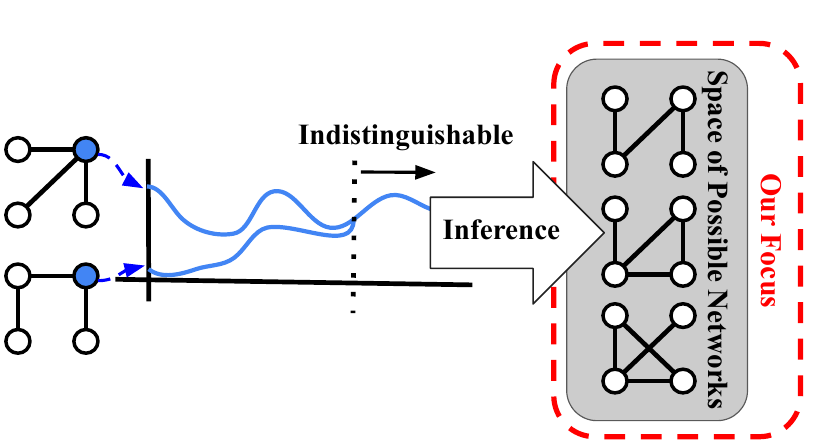}
      \caption{Various network structures can have observed trajectories that contract to one another making them indistinguishable from each other. We derive sufficient conditions on the dynamics for this phenomenon to occur, and in doing so, we are able to generate other candidate structures that behave similarly (i.e., have the same observed trajectories).}
      \label{fig:overview}
   \end{figure}
   
We begin by developing a framework to analyze the space of networks that are indistinguishable based on partial measurements of the nodes' states, and establish sufficient conditions for indistinguishability (Section \ref{sec:Contraction_theory}). We then show how this framework can be applied to nonlinear (Section \ref{sec:nonlin_nets}) and linear (Section \ref{sec:linear_nets}) networked systems. We apply this theory to Kuramoto oscillator networks, and show how one can generate a variety of candidate Kuramoto oscillator networks that are indistinguishable based on partial measurements (Section \ref{sec:Kuramoto}, \ref{sec:sims}). Concluding remarks are provided in Section \ref{sec:conclusion}.

\noindent\textbf{Notation.~} We denote the vector of all 1's with $\mathds{1}$. The norm with respect to a matrix $M$ is $\| x\|_M = \sqrt{x^\top M x}$, and the $\ell_\infty$ norm is $\| x\|_\infty = \max_i |x_i|$. The partial derivative of a function $f$ with respect to a variable $s$ is denoted by $\partial_sf$ and the time derivative by $\frac{df}{dt}$ or equivalently $\dot{f}$. The pseudoinverse of a matrix $A\in \R^{m\times n}$ with $\text{rank}(A) = m$ is $A^\dagger = A^\top (AA^\top)^{-1}$. The spectral abscissa of a matrix is given as $\alpha(A) = \max_i\{ \text{Re}(\lambda_i(A))\}$. Lastly, $\range(\cdot)$ and $\nullsp(\cdot)$ correspond to the range and nullspace of their respective arguments, and $\oplus$ is the  direct sum of two vector spaces. 

\section{Contraction Theory for System Comparisons}\label{sec:Contraction_theory}

We consider two systems $\Sigma$ and $\tilde\Sigma$ evolving under the following dynamics 
\begin{equation}\label{sys:original}
    \Sigma : \begin{cases}
        \dx = f(x),  & x(0) = x_0\\
        y = Cx,
    \end{cases}
\end{equation}

\begin{equation} \label{sys:perturbed}
        \tilde{\Sigma} : \begin{cases}
        \dtx = g(\tx), & \tx(0) = \tx_0\\
        \ty = C\tx.
    \end{cases}
\end{equation}
Here $x, \tx \in \R^n$ encode the states of each of the systems and, $f: \R^n \rightarrow \R^n$ and  $g: \R^n \rightarrow \R^n$ are nonlinear functions encoding the dynamics. Lastly, $y, \ty \in \R^p$ are measurements of the systems through the observation matrix $C\in \R^{p\times n}$. 
The choice to analyze linear measurements of the states is reasonable as, for instance, in neuroscience settings microelectrode arrays are used to capture single neuron or multiple neuron recordings in a linear superposition. We now formalize the notion of indistinguishability. 
\begin{definition}
    We say that the systems \eqref{sys:original} and \eqref{sys:perturbed} \textit{become indistinguishable} with respect to some positive definite matrix $M$ if $\forall \epsilon >0 ~~\exists T > 0$ such that $\|y(t) - \tilde y(t)  \|_M <\epsilon ~~ \forall t > T$.  
\end{definition}
To compare the trajectories of the systems \eqref{sys:original} and \eqref{sys:perturbed} and determine whether they become indistinguishable, we consider the auxiliary system inspired by \cite{JOUFFROY_2005}:
\begin{equation}\label{sys:aux}
    \Sigma_{\text{aux}} : \begin{cases}
        \dxa = f(\xa) + s\Delta(\xa), & \xa(0) = (1-s)x_0 + s\tx_0 \\
        \ya = C\xa, \\
        \Delta(\xa) = g(\xa) - f(\xa).
    \end{cases}
\end{equation}
Here $s\in [0,1]$, and for $s = 0$ we recover the dynamics of $\Sigma$ and for $s=1$ we recover the dynamics of $\tilde{\Sigma}$.
\begin{remark}
    The auxiliary system defined in \eqref{sys:aux} is a convex combination of the dynamics in \eqref{sys:original} and \eqref{sys:perturbed}. In \cite[Ex. 2.2]{WANG_2004} a convex combination of contracting systems is shown to tend to a shared trajectory or equilibrium point.
\end{remark}
We now measure the time-varying length \cite{JOUFFROY_2005, SIMPSONPORCO_2014} between the trajectories of \eqref{sys:original} and \eqref{sys:perturbed} with respect to some positive definite matrix $M(\xa) \in \R^{p\times p}$  defined as
\begin{equation}\label{eq:length}
    L(t)  = \int_{0}^{1} \| \partial_s \ya (t,s)\|_M ds.
\end{equation}
We now study the dynamics of this length. For brevity, we define the matrix $M_C = C^\top MC$.

\begin{lemma}\label{lem:length_diff_ineq}
    Assume the auxiliary system \eqref{sys:aux} is contracting with respect to $M_C$ with rate $\mu\in \R_{>0}$, i.e., it satisfies the linear matrix inequality (LMI)
    \begin{multline}\label{eq:LMI_contracting}
    \big(\tfrac{\partial f}{\partial \xa}+ s\tfrac{\partial \Delta}{\partial \xa}\big)^\top M_C +\dot M_C+ M_C \big(\tfrac{\partial f}{\partial \xa}+ s\tfrac{\partial \Delta}{\partial \xa}\big)  \preceq -2\mu M_C,
\end{multline}
    then the time-varying length $L(t)$ satisfies the differential inequality
    \begin{equation}\label{eq:length_diff}
        \dot{L}(t) \leq -\mu L(t) + \int_0^1 \| C\Delta(\xa(t,s))\|_Mds.
    \end{equation}
\end{lemma}
\begin{proof}
    Identify that $\ya(t,s)$ can be implicitly expressed as
\begin{equation}
    \ya(t,s) = C\Big( \xa(0,s) + \int_{\tau=0}^t f(\xa(\tau, s)) + s\Delta(\xa(\tau,s)) d\tau\Big).
\end{equation}
Differentiating with respect to $s$ yields
\begin{multline}
        \partial_s \ya(t,s) = \\ C\Big( \partial_s\xa(0,s) + \int_{\tau=0}^t \partial_s\{f(\xa(\tau, s)) + s\Delta(\xa(\tau,s))\} d\tau\Big).
\end{multline}
Finally, differentiating with respect to time leads to 
\begin{equation}
        \tfrac{d}{dt}\partial_s \ya(t,s) =  C \tfrac{\partial f}{\partial \xa}\partial_s\xa + C\Delta + s C\tfrac{\partial \Delta}{\partial \xa}\partial_s\xa.
\end{equation}
This implies that
\begin{multline}
        \tfrac{d}{dt} \| \partial_s \ya \|_M^2  = \partial_s\xa^\top \bigg\{\tfrac{\partial f}{\partial \xa}^\top M_C +\dot M_C+ M_C \tfrac{\partial f}{\partial \xa}  \\+ s\Big(\tfrac{\partial \Delta}{\partial \xa}^\top M_C + M_C \tfrac{\partial \Delta}{\partial \xa}\Big)\bigg\}\partial_s\xa + 2
        \partial_s\xa^\top M_C\Delta .
\end{multline}
By employing the assumption \eqref{eq:LMI_contracting} and applying the Cauchy-Schwarz inequality we arrive at
\begin{equation}\label{eq:ys_ineq}
    \tfrac{d}{dt} \| \partial_s \ya \|_M^2 \leq -2\mu\| \partial_s \ya \|_M^2 + 2\| \partial_s \ya  \|_M\|C\Delta\|_M.
\end{equation}
Differentiating \eqref{eq:length} and substituting \eqref{eq:ys_ineq} will lead to the result (see \cite[Thm. 2.3]{SIMPSONPORCO_2014} for details).
\end{proof}

\begin{corollary}
    If we find a semidefinite matrix $M \succeq 0$ satisfying \eqref{eq:LMI_contracting} we can guarantee contraction in all space except for $\nullsp(M)$. See \cite{FB-CTDS} for discussion on the notion of semicontraction. 
\end{corollary}

Lemma \ref{lem:length_diff_ineq} establishes that $L(t)$ decreases and remains bounded provided $\Delta(\xa)$ is bounded $\forall t, s$. This result is analogous to the discussion on disturbed flows in \cite{LOHMILLER_1998}. 

\begin{remark}
    Lemma \ref{lem:length_diff_ineq} is an extension of traditional contraction theory results. This framework shows the impact of partial observation via $M_C$, and provides an explicit approach on how to compare the dynamics across systems.  
\end{remark}

\begin{corollary}\label{cor:Cdel_0}
    If the difference between the dynamics $\Delta(\xa)$ satisfies $C\Delta(\xa) = 0$ the two systems' measurements contract to each other with rate $\mu$. 
\end{corollary}
\begin{proof}
    Assume that $C\Delta(\xa) = 0$ then \eqref{eq:length_diff} becomes
    \begin{equation}
    \dot L(t) \leq -\mu L(t).
    \end{equation}
    Applying Gr\"{o}nwall's inequality yields
    \begin{equation}
        L(t) \leq e^{-\mu t}L(0).
    \end{equation}
    Hence the length contracts with rate $\mu$ as required.
\end{proof}

Corollary \ref{cor:Cdel_0} establishes that if the difference between the dynamics $\Delta(\xa)$ is in the nullspace of the observation matrix $C$ then as $t \rightarrow \infty$ the observations $y(t)$ and $\ty(t)$ will become indistinguishable. This motivates determining the conditions that need to be enforced on the Jacobians of $f$ and $\Delta$ in order for the LMI in \eqref{eq:LMI_contracting} to be satisfied; which would imply that the systems in \eqref{sys:original} and \eqref{sys:perturbed} would become indistinguishable. To do so, we consider constant matrices $M(x) = M$ and establish the following definition. 

\begin{definition}
    We call a matrix $A$ \textit{contractive in the observable space} defined by $C$ if
    \begin{equation}\label{eq:contraction_obsv_sp}
        \alpha(CAC^\dagger) < 0.
    \end{equation}
\end{definition}

We now show that contraction in the observable space is directly related to the indistinguishability of network structures by showing that it is a requirement for \eqref{eq:LMI_contracting} to be satisfied. In the context of the auxiliary system, we  require 
\begin{equation}\label{eq:condition_jacobian_general}
    \alpha\Bigl(C \Bigl(\frac{\partial f}{\partial \xa} + s \frac{\partial \Delta}{\partial \xa}\Bigr)C^\dagger\Bigr) < 0.
\end{equation}

\begin{theorem}\label{thm:LMI_existence}
    Let $A$ be an arbitrary square matrix and suppose $C$ has full row rank, and further $\nullsp (C) \subseteq \nullsp (CA)$. Then, 
    $\alpha(CAC^\dagger) \leq -\mu$ if and only if $\exists M \succ 0$ such that
    \begin{equation}\label{eq:LMI_A}
        A^\top M_C + M_C A \preceq -2\mu M_C.
    \end{equation}
\end{theorem}

\begin{proof}
    We begin with the forward implication. 
    Suppose that $\alpha(CAC^\dagger) \leq -\mu$, then $\exists M \succ 0$ such that the following LMI holds
    \begin{equation}
        (CAC^\dagger)^\top M + M (CAC^\dagger) + 2\mu M \preceq 0.
    \end{equation}
    Pre- and post-multiplying by $CC^\top$ we arrive at
    \begin{equation}\label{eq:C_semidef}
        C(A^\top M_C + M_CA  +2\mu M_C)C^\top \preceq 0.
    \end{equation}
    Identifying that $\R^n = \range(C^\top) \oplus \nullsp(C)$ we can decompose any vector $v$ as $v= C^\top m + q$ where $q \in \nullsp(C)$. Hence,
    \begin{equation}\begin{aligned}
                & v^\top ( A^\top M_C + M_CA  +2\mu M_C) v \\
                &= ( C^\top m + q)^\top ( A^\top M_C + M_CA  +2\mu M_C) ( C^\top m + q)\\
                &= \underbrace{m^\top C(A^\top M_C + M_CA  +2\mu M_C)C^\top m}_{\leq 0 \text{ by } \eqref{eq:C_semidef}} \\
                &+ \underbrace{q^\top (A^\top M_C + M_CA  +2\mu M_C)q}_{= 0 \text{ since } q \in \nullsp (C)} \\
                &+ 2 m^\top C( A^\top M_C + M_CA  +2\mu M_C)q \\ 
                &\leq 2 m^\top C( A^\top M_C + M_CA  +2\mu M_C)q\\
                &= 2m^\top CM_CAq = 2m^\top CC^\top M CA q = 0. 
    \end{aligned}
    \end{equation}
    By assumption since $CAq = 0$, we arrive at the result 
    \begin{equation}
        A^\top M_C + M_C A \preceq -2\mu M_C.
    \end{equation}

We now prove the reverse implication by assuming that \eqref{eq:LMI_A} holds. 
Pre- and post-multiplying \eqref{eq:LMI_A} by $C^\top C \succeq 0$ we arrive at
    \begin{equation}
        C^\top(CA^\top M_CC^\top + CM_CAC^\top + 2\mu CM_CC^\top)C \preceq 0.
    \end{equation}
    Since $\R^p = \range(C) \oplus \nullsp(C^\top)$ we can express any vector $v$ as $v = Cm + q$ where $q \in \nullsp(C^\top)$. Applying a similar argument as made previously
    \begin{equation}
        v^\top (CA^\top M_CC^\top + CM_CAC^\top + 2\mu CM_CC^\top)v \leq 0 ~~ \forall v,
    \end{equation}
    which implies
    \begin{equation}
        CA^\top M_CC^\top + CM_CAC^\top + 2\mu CM_CC^\top \preceq 0.
    \end{equation}
    Now, pre- and post-multiplying by $(CC^\top)^{-1}$ establishes
    \begin{equation}
        (CAC^\dagger)^\top M + M (CAC^\dagger) + 2\mu M \preceq 0.
    \end{equation}
    Hence we conclude $\alpha (CAC^\dagger) \leq - \mu$ as required.
\end{proof}

\begin{remark}
    The requirement $\alpha (CAC^\dagger) \leq - \mu$ is essentially a semicontractivity requirement on $A$ \cite{FB-CTDS}. 
\end{remark}

The assumption $\nullsp (C) \subseteq \nullsp (CA)$ is not required for the reverse implication in Theorem \ref{thm:LMI_existence}. Moreover, this assumption amounts to requiring $\nullsp (C)$ to be $A$-invariant. In practice this is not regularly satisfied, but in certain cases (see Section \ref{sec:Kuramoto}) can be. When $C = I$ it is trivially satisfied,  \eqref{eq:LMI_A} reduces to a Lyapunov equation and \eqref{eq:contraction_obsv_sp} implies $A$ must be Hurwitz.

Theorem \ref{thm:LMI_existence} implies that satisfying $\alpha(CAC^\dagger) < 0$ is essential in order for the matrix $M$ to exist. This underscores that the auxiliary system in \eqref{sys:aux} must be contracting in the observable space for the two systems being compared to have trajectories that contract to each other and become indistinguishable. That is, it is necessary to identify whether \eqref{eq:condition_jacobian_general} is satisfied which may be challenging if the Jacobians depend explicitly on $\xa$.

As a summary, the following steps will determine whether two systems are indistinguishable: 
\begin{enumerate}[(i)]
    \item Verify \eqref{eq:condition_jacobian_general} and find a suitable $M \succeq 0$ that satisfies \eqref{eq:LMI_contracting} by applying Theorem \ref{thm:LMI_existence}. Semidefinite matrices allow for further generality.
    \item To satisfy Corollary \ref{cor:Cdel_0}, determine whether $C\Delta(\xa) = 0$.
\end{enumerate}

\subsection{Application I: Nonlinear Networked Systems}\label{sec:nonlin_nets}

The results described previously provide a general framework for comparing the observed dynamics of nonlinear systems. We now particularize these results to nonlinear networked systems to describe how a variety of network structures can yield dynamics that are similar. Specifically, consider the dynamics of a network given as 
\begin{equation}\label{sys:nonlinear_nets}
    \Sigma_{\text{net}}:\begin{cases}
        \dx = h(x)  + W \Phi(x),& x(0) = x_0\\
        y = Cx.
    \end{cases}
\end{equation}
Here $h(x)$ is a vector encoding the self dynamics of individual nodes in the network, i.e., it does not encode any interaction between nodes. Interactions between nodes are captured in the coupling term $W\Phi(x)$, where $W$ encodes the structure of the network, and $\Phi(x)$ is a nonlinear function applied element-wise to the nodal states. Applying our framework, we can determine whether a perturbed network $W+ \Delta$ will yield identical or similar dynamics. The corresponding perturbed system dynamics would be
\begin{equation}
    \tilde\Sigma_{\text{net}}:\begin{cases}
        \dx = h(x)  + (W+\Delta) \Phi(x),& x(0) = x_0\\
        y = Cx.
    \end{cases}
\end{equation}
The condition outlined in \eqref{eq:condition_jacobian_general} amounts to verifying whether
\begin{equation}
    \alpha\Bigl(C \Bigl(\frac{\partial h}{\partial \xa} + (s\Delta + W) \frac{\partial \Phi}{\partial \xa}\Bigr)C^\dagger\Bigr) < 0.
\end{equation}

\subsection{Application II: Linear Systems}\label{sec:linear_nets}
To build some intuition, we now consider linear systems. A complete and thorough analysis, via observability, of the linear setting can be found in \cite{GILL_2025}. Specifically, consider a linear system given as
\begin{equation}\label{sys:linear}
    \Sigma_{\text{lin}}:\begin{cases}
        \dx = Ax, & x(0) = x_0\\
        y = Cx.
    \end{cases}
\end{equation}
Suppose we perturb the system matrix $A$ to evolve under $A+ \Delta$ with $C\Delta = 0$. Then to satisfy \eqref{eq:LMI_contracting} with a constant matrix $M(x) = M$, we need 
\begin{equation}
    A^\top M_C + M_C A \preceq -2\mu M_C.
\end{equation}
From Theorem \ref{thm:LMI_existence} we require $\alpha(CAC^\dagger) \leq - \mu $ and $\nullsp(C) \subseteq \nullsp(CA)$. Consider the following example:  
\begin{equation}
    A = \begin{bmatrix}
        -1 & 0 \\ -1 & 1 
    \end{bmatrix}, ~~ C = \begin{bmatrix}
        1 & 0 
    \end{bmatrix}. 
\end{equation}
Notice with this choice of $A$ and $C$, $\alpha(CAC^\dagger) = \alpha(-1) = -1 < 0$ and $\nullsp(CA) = \nullsp(C)$. Moreover, any $\Delta$ such that $C\Delta = 0$ is viable, this amounts to perturbations of the form
\begin{equation}
    \Delta = \begin{bmatrix}
        0 & 0  \\ \Delta_{21} & \Delta_{22}
    \end{bmatrix}.
\end{equation}
This result is fully consistent with the fact that $x_1$ is fully decoupled from $x_2$ and hence the dynamics of $x_2$ can be modified freely, as expected. 

\section{Kuramoto Oscillator Networks}\label{sec:Kuramoto}

We now focus our attention on the Kuramoto model of coupled oscillators, which is often used in modeling neural systems \cite{STROGATZ_1994}. We study the impact of this model's contracting dynamics on its identifiability. We apply our framework to prove that a variety of network structures are indistinguishable from each other. As in \cite{JADBABAIE_2004, ZHU_2020} we express the Kuramoto oscillator network dynamics as
\begin{equation}\label{eq:kuramoto_system}
\Sigma_{\text{Kuramoto}}:\begin{cases}
        \dx = \omega - B \mathcal{A}\sin(B^\top x), \\
        y = C x.
\end{cases}
\end{equation}
Let $e$ be the number of edges in the network, then $B \in \R^{n \times e}$ is the incidence matrix encoding the network structure with arbitrary assignment of sink and source nodes, and $\mathcal{A} \in \R^{e\times e}$ is a diagonal matrix of the edge weights (not to be confused with the adjacency matrix). We always define $B$ to encode the incidence matrix for a complete graph with $n$ nodes, and hence, if an edge is not present we set the corresponding edge weight encoded in $\mathcal{A}$ to be zero. Here $\omega$ encodes the non-identical natural frequencies of the oscillators. 

\begin{remark}
    In the case of fully connected networks of oscillators with identical natural frequencies, all oscillators would synchronize \cite{DORFLER_2014} to the common natural frequency $\omega$ and all connected networks would be indistinguishable as all observations would differ by at most a constant phase shift. 
\end{remark}

We now consider how to analyze how variations to the network structure alter the measured dynamics of the oscillators. We define $\delta \in \R^e$ to be a perturbation to the edge weights, then the corresponding auxiliary system is
\begin{equation}\label{eq:kuramoto_auxiliary}
    \begin{cases}
        \dxa = \omega - B (s\text{diag}(\delta)+\mathcal{A})\sin(B^\top \xa), \\
        \ya = C \xa.
\end{cases}
\end{equation}
Further, as in \eqref{sys:aux} we define 
\begin{subequations}\begin{align}
    &f(\xa) = - B \mathcal{A}\sin(B^\top \xa), \\
    &\Delta(\xa) = - B \text{diag}(\delta)\sin(B^\top\xa).
\end{align}
\end{subequations}
The Jacobians of  $f(\xa)$  and $\Delta(\xa)$ are
\begin{subequations}\label{eq:jacobians}\begin{align}
        & \frac{\partial f}{\partial \xa} = -B\mathcal{A}\text{diag}(\cos(B^\top\xa))B^\top, \\
    & \frac{\partial \Delta}{\partial \xa} = -B \text{diag}(\delta)\text{diag}(\cos(B^\top\xa))B^\top.
\end{align}
\end{subequations}
The Jacobians \eqref{eq:jacobians} are weighted Laplacian matrices and hence $\mathds{1} \in \nullsp(\frac{\partial f}{\partial \xa})$ and  $\mathds{1} \in \nullsp(\frac{\partial \Delta}{\partial \xa})$, thus in the case of full observations $C=I$ we can at most expect semicontraction (i.e., contraction in the space orthogonal to $\text{span}(\mathds{1})$). The second smallest eigenvalue of the auxiliary systems' Jacobian will always be greater than 0 if and only if
\begin{equation}\label{eq:joint_jac}
    -B (s\text{diag}(\delta) + \mathcal{A})\text{diag}(\cos(B^\top\xa))B^\top
\end{equation}
encodes the weighted Laplacian of a connected graph \cite{JADBABAIE_2004}. Then, \eqref{eq:condition_jacobian_general} amounts to determining whether
\begin{equation}
    \alpha(-CB (s\text{diag}(\delta) + \mathcal{A})\text{diag}(\cos(B^\top\xa))B^\top C^\dagger) < 0.
\end{equation}

\begin{remark}
    The conditions for two different oscillator networks to contract to each other are analogous to requirements for synchrony found in  \cite{JADBABAIE_2004, DELABAYS_2023, DORFLER_2014}. This reiterates previous findings in \cite{CHEN_2009} that establish a link between synchrony and the inability to identify network structures.  
\end{remark}

We now apply our framework to find other indistinguishable networks when starting with a candidate network and corresponding measurement scheme. We consider an example comprised of four oscillators which is displayed in Fig.~\ref{fig:kuramoto_partial} (Net 1). We define the incidence matrix for the complete graph as
\begin{equation}\label{eq:B_matrix}
    B = \begin{bmatrix}
        1 & 1 & 1 & 0 & 0 & 0 \\
        -1 & 0 & 0 & 1 & 1 & 0 \\
        0 & -1 & 0 & -1 & 0 & 1 \\
        0 & 0 & -1 & 0 & -1 & -1
    \end{bmatrix}
\end{equation}
and the corresponding edge weight matrix as
\begin{equation}\label{eq:A_matrix}
    \mathcal{A} = \text{diag}(\begin{bmatrix}
        a_{12} & a_{13} & 0 & 0 & a_{24} & a_{34}
    \end{bmatrix}^\top).
\end{equation}
This network is connected with each oscillator connected to its nearest neighbor. We consider non-identical natural frequencies given as 
\begin{equation}
    \omega = \begin{bmatrix}
        \omega_1 & \omega_2 & \omega_3 & \omega_4
    \end{bmatrix}^\top.
\end{equation}
The corresponding differential equations describing the oscillator dynamics are:
\begin{subequations}\begin{align}
    & \dx_1 = \omega_1 - a_{12}\sin(x_1-x_2) - a_{13}\sin(x_1-x_3) \\ 
    & \dx_2 = \omega_2 - a_{12}\sin(x_2-x_1) - a_{24}\sin(x_2-x_4) \\
    & \dx_3 = \omega_3 - a_{13}\sin(x_3-x_1) - a_{34}\sin(x_3-x_4) \\ 
    & \dx_4 = \omega_4 - a_{24}\sin(x_4-x_2) - a_{34}\sin(x_4-x_3).
\end{align}
\end{subequations}
Now consider a measurement scheme of the form
\begin{equation}\label{eq:C_mat_kuramoto}
    C = \begin{bmatrix}
        1 & 1 & 0 & 0 \\
        0 & 0 & 1 & 1
    \end{bmatrix}.
\end{equation}
The measurement scheme defined in \eqref{eq:C_mat_kuramoto} indicates that we are measuring neighboring nodes in an averaging fashion. Now, consider new networks with an edge weight matrix
\begin{equation}
    \mathcal{A} + \text{diag}(\begin{bmatrix}
        \delta_1 & \delta_2 & \delta_3 & \delta_4& \delta_5&\delta_6 
    \end{bmatrix}^\top).
\end{equation}

\begin{figure}[thpb]
      \centering
      \includegraphics[clip, trim=1.8cm 4.05cm 2cm 5cm,width = \linewidth]{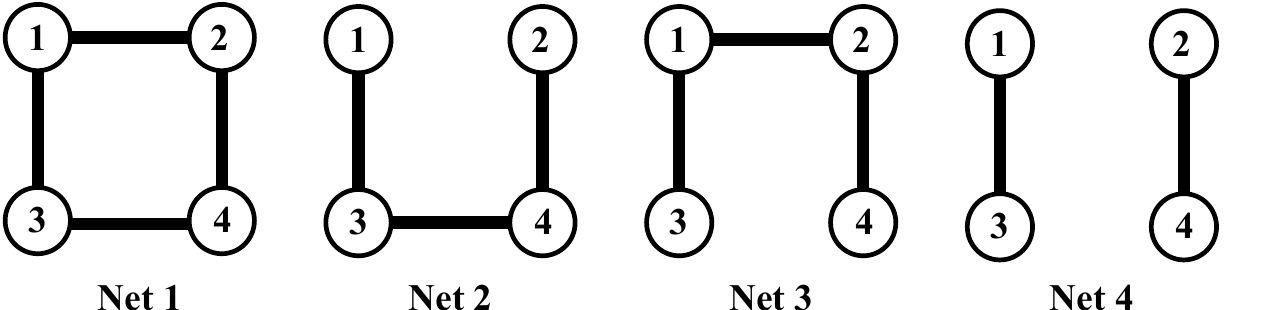}
      \caption{We display the topology encoded in \eqref{eq:B_matrix} and \eqref{eq:A_matrix} (Net 1). The three other topologies are equivalent based on the analysis when an average of nodes 1 and 2 and an average of nodes 3 and 4 are measured.}
      \label{fig:kuramoto_partial}
\end{figure}

In order to establish (semi)contractivity of each of the systems' observed trajectories to one another, we need to (i) find a matrix $M$ satisfying \eqref{eq:LMI_contracting} and (ii) show $C\Delta(\xa) = 0$ to satisfy Corollary \ref{cor:Cdel_0}. Notice that by setting $\delta_2 = \delta_3 = \delta_4 = \delta_5 = 0$, simple matrix multiplication will yield that 
\begin{equation}
    C\Delta(\xa) = -CB \text{diag}(\delta)\sin(B^\top \xa) = 0,
\end{equation}
which confirms (ii). 
We now consider satisfying (i), evaluating \eqref{eq:condition_jacobian_general}, i.e., $C\big(\frac{\partial f}{\partial \xa} + s \frac{\partial \Delta}{\partial \xa}\big) C^\dagger$ we arrive at
\begin{multline}\label{eq:partial_jacobian_example}
    \frac{1}{2}(a_{13} \cos(\xa_1 -\xa_3) + a_{24}\cos(\xa_2 -\xa_4))\begin{bmatrix}
        -1& 1\\ 1 & -1 
    \end{bmatrix}.
\end{multline}
Note that \eqref{eq:partial_jacobian_example} has one negative eigenvalue and the nullspace is simply equal to $\text{span}(\mathds{1})$ provided that
\begin{equation}\label{eq:positivity_edge}
    \frac{1}{2}(a_{13} \cos(\xa_1 -\xa_3) + a_{24}\cos(\xa_2 -\xa_4)) > 0.
\end{equation}
Hence we may conclude, that the observations will be contracting in all space except $\text{span}(\mathds{1})$. Also, if \eqref{eq:positivity_edge} is satisfied, then if we can find a matrix $M$ to satisfy \eqref{eq:kuramoto_example_lmi}, this matrix will also satisfy \eqref{eq:LMI_contracting}, as required.

\begin{equation} \label{eq:kuramoto_example_lmi}
    \begin{bmatrix}
        -1& 1\\ 1 & -1 
    \end{bmatrix}^\top M + M \begin{bmatrix}
        -1& 1\\ 1 & -1 
    \end{bmatrix} = -4M .
\end{equation}

Specifically, with $M = I - \frac{1}{n}\mathds{1} \mathds{1}^\top$ as in \cite{DELABAYS_2023}, \eqref{eq:kuramoto_example_lmi} (and therefore \eqref{eq:LMI_contracting}) is satisfied. To verify $\nullsp(C) \subseteq \nullsp(C(\frac{\partial f}{\partial \xa} + s \frac{\partial \Delta}{\partial \xa}))$, it is sufficient to satisfy
\begin{equation}\label{eq:condition_nullspace}
    a_{13} \cos(\xa_1 -\xa_3) = a_{24}\cos(\xa_2 -\xa_4).
\end{equation}
If \eqref{eq:condition_nullspace} is satisfied following Theorem \ref{thm:LMI_existence} we can use the suggested matrix $M$ with the caveat that we cannot achieve contractivity in the space defined by $\text{span}(\mathds{1})$, and thus we have satisfied requirement (i).
Furthermore, since there were no constraints enforced on the remaining perturbation weights, we may arbitrarily alter $\delta_1$ and $\delta_6$, this implies that all networks shown in Fig.~\ref{fig:kuramoto_partial} display observations that differ by only a constant phase shift. 

The conditions outlined in \eqref{eq:positivity_edge} and \eqref{eq:condition_nullspace} are challenging to verify due to the dependence on the state $\xa$ and the parametrization variable $s$. However, we now outline alternative sufficient conditions for satisfying \eqref{eq:positivity_edge} and \eqref{eq:condition_nullspace}. We introduce the following definition similar to those in \cite{DORFLER_2014, DELABAYS_2023}.

\begin{definition}
    We say that the oscillator states $\xa$ are $\gamma$-\textit{phase cohesive} if $\xa\in \Theta(\gamma) = \{ \xa : \|B^\top \xa\|_\infty < \gamma \} ~ \forall t \geq 0$. 
\end{definition}

Notably, if the oscillator states are $\pi/2$-phase cohesive, i.e., $\xa \in \Theta(\frac{\pi}{2})~ \forall t \geq 0$, then \eqref{eq:positivity_edge} is satisfied for any positive edge weights $a_{13}$ and $a_{24}$. Remaining phase cohesive has been extensively studied as it has a close relationship to achieving synchrony, we point the reader to \cite{DORFLER_2014} for a thorough discussion. Typically, the conditions to remain phase cohesive are conservative estimates on the maximum difference between the oscillator frequencies $\omega$, for instance, one such condition from \cite{DORFLER_2014} is for
\begin{equation}
    \lambda_2(B(s\text{diag}(\delta) + \mathcal{A})B^\top) > \|B^\top \omega \|_2 .
\end{equation}
This condition would not be directly applicable to the disconnected topology in Fig.~\ref{fig:kuramoto_partial} (Net 4) as $\lambda_2(B(\text{diag}(\delta) + \mathcal{A})B^\top) = 0$. However, in Section \ref{sec:sims} we observe that this network is still indistinguishable from the others. 

To satisfy \eqref{eq:condition_nullspace} it suffices to ensure that there is some symmetry, in the sense that permuting the pair of labels 1 and 3 with 2 and 4 will lead to identical difference dynamics $\xa_1 - \xa_3$ and $\xa_2-\xa_4$ for each of the oscillator pairs. This is trivially achievable if all the oscillator frequencies and edge weights are identical; if not, certain selections of similar frequencies and edge weights are sufficient. 

\section{Simulations}\label{sec:sims}

To validate our findings we simulate all four networks from Fig.~\ref{fig:kuramoto_partial} and plot the corresponding measurements $y_1(t)$ and $y_2(t)$ in Fig.~\ref{fig:kuramoto_partial_sim}.
For Fig.~\ref{fig:kuramoto_partial} (Right), we select initial conditions for nodes 1 and 2 randomly. In order to satisfy \eqref{eq:condition_nullspace}, node 4 has the same initial condition as node 1; node 3 has the same initial condition as node 2.

\begin{figure}[thpb]
      \centering
      \includegraphics[clip, trim=0cm 0cm 0cm 0cm,width = \linewidth]{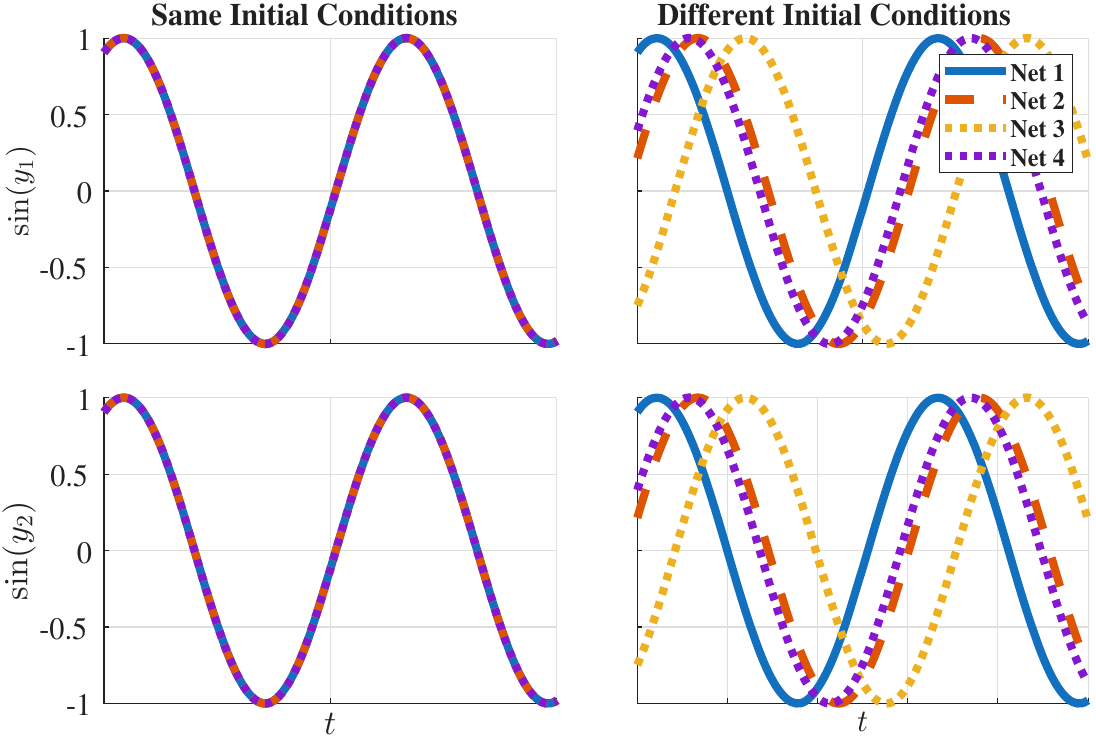}
      \caption{Comparison of different Kuramoto networks displayed in Fig.~\ref{fig:kuramoto_partial}. When beginning with the same initial condition (Left column), all four network structures' measurements are identical to each other over the entire measurement period and have no phase shifts between them. When the networks start with different initial conditions (Right column) the states differ by a constant phase shift.}
      \label{fig:kuramoto_partial_sim}
      \vspace{-1em}
\end{figure}

We see, as expected, in Fig.~\ref{fig:kuramoto_partial_sim} that when the networks are started with all the same initial conditions (Fig.~\ref{fig:kuramoto_partial_sim} (Left)) they do not display any phase shift between them, and are identical for all time, implying that the networks are all indistinguishable from one another. When the initial conditions are different (Fig.~\ref{fig:kuramoto_partial_sim} (Right)), the networks have measurements that share a common waveform and only differ by a fixed phase shift. Despite differing by a phase shift the networks are still indistinguishable as measuring the behavior past the period of transient dynamics the functional form of the measurements are identical. Remarkably, in the partial measurement setting it is possible to mischaracterize connected networks as disconnected networks (Fig.~\ref{fig:kuramoto_partial} (Net 4)) as illustrated by this example. 

Displayed in Fig.~\ref{fig:kuramoto_random} the oscillators were simulated with random natural frequencies $\omega$ and random initial conditions $x_0$ which do not necessarily adhere to the requirement \eqref{eq:condition_nullspace}. Despite this, the networks still observe similar behavior as in Fig. \ref{fig:kuramoto_partial_sim}. This emphasizes that the conditions established are sufficient but not necessary. 

\begin{figure}[thpb]
      \centering
      \includegraphics[clip, trim=0cm 0cm 0cm 0cm,width = \linewidth]{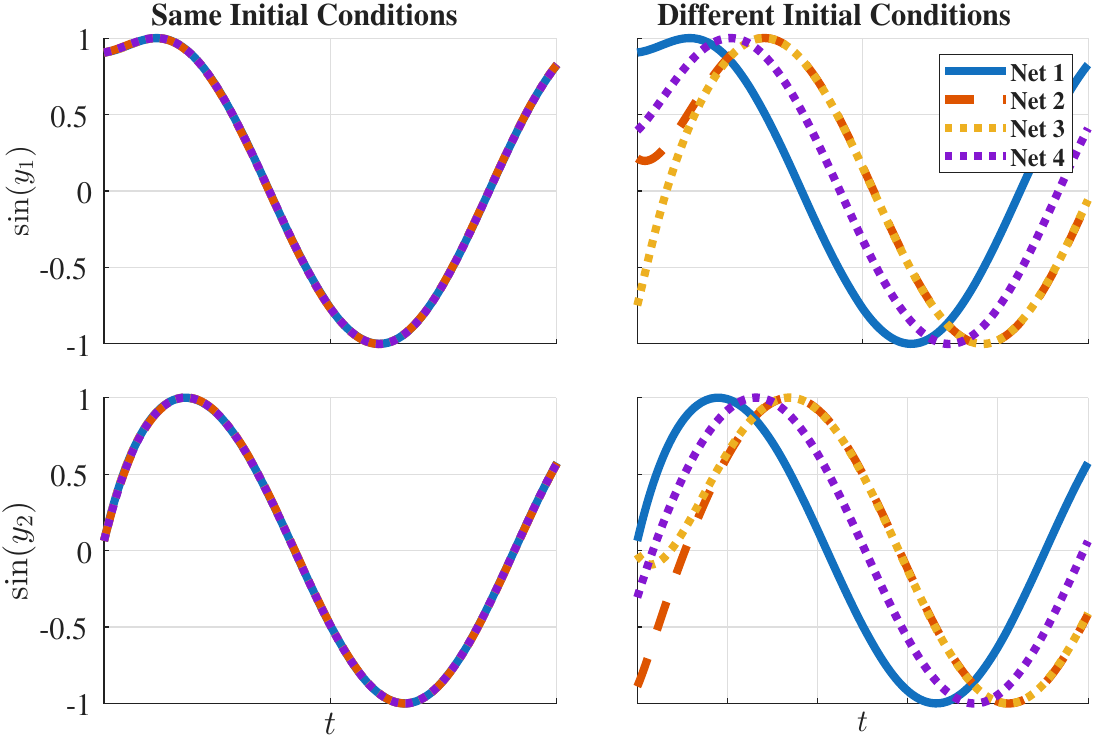}
      \caption{Comparison of different Kuramoto networks displayed in Fig.~\ref{fig:kuramoto_partial}. In this case, $x_0$ and $\omega$ are entirely random and thus violate \eqref{eq:condition_nullspace}. Networks are simulated all with the same initial condition (Left column) and with different initial conditions (Right column).}
      \label{fig:kuramoto_random}
\end{figure}

In summary, in order to consider whether other candidate network structure are indistinguishable from a hypothesis network it suffices to establish satisfaction of requirements (i) and (ii) (i.e., applying Theorem \ref{thm:LMI_existence} and Corollary \ref{cor:Cdel_0}) outlined at the end of Section \ref{sec:Contraction_theory}.



\section{Conclusions and Future Work}\label{sec:conclusion}
In this paper we have established a series of results that describe how contractivity (in a single dynamical system and in the observable space) can result in challenges in the identification of network topologies. We provide sufficient conditions for observed trajectories of different nonlinear systems to contract to one another by leveraging the contraction theory framework and establishing a notion of contractivity in the observable space. Particularizing our results to Kuramoto oscillator networks we established that synchrony, phase cohesiveness, and symmetry play a key role in determining whether different topologies are indistinguishable from one another.
We anticipate the framework proposed will find useful applications to various other models of neural dynamics in order to study how different implementations of brain networks or neural circuits exhibit the same function and can help guide new techniques for the inference of the structures in these network models.

\bibliographystyle{IEEEtran}

\bibliography{IEEEabrv,IEEEexample}

\addtolength{\textheight}{-12cm}   







\end{document}